\documentclass[a4paper,UKenglish,cleveref, autoref]{article}
\usepackage{tikz}
\usepackage[algoruled,linesnumbered,vlined]{algorithm2e}
\usepackage{amsthm}
\usepackage{amsmath}
\usepackage{amsfonts}
\usepackage{url}
\usepackage{hyperref}
\usepackage{geometry}
\usepackage{microtype}
\usepackage{graphicx}

\newtheorem{lemma}{Lemma}
\newtheorem{corollary}{Corollary}
\newtheorem{theorem}{Theorem}
\newtheorem{definition}{Definition}
\newtheorem{remark}{Remark}

\title{Non-Rectangular Convolutions and (Sub-)Cadences with Three Elements}

\author{Mitsuru Funakoshi\thanks{Department of Informatics, Kyushu University, Japan. Email: mitsuru.funakoshi@inf.kyushu-u.ac.jp} \and Julian Pape-Lange\thanks{Technische Universit\"at Chemnitz, Stra\ss e der Nationen 62, 09111 Chemnitz, Germany. Email: julian.pape-lange@informatik.tu-chemnitz.de}}
\date{}

\begin{document}

\maketitle

\begin{abstract}
The discrete acyclic convolution computes the $2n-1$ sums \[\displaystyle\sum_{\substack{i+j=k\\ (i,j)\in [0,1,2,\dots,n-1]^2}} a_i b_j\]
in $\mathcal{O}\left(n\log n\right)$ time. By using suitable offsets and setting some of the variables to zero, this method provides a tool to calculate all non-zero sums \[\displaystyle\sum_{\substack{i+j=k\\ (i,j)\in P\cap \mathbb{Z}^2}} a_i b_j\]
in a rectangle $P$ with perimeter $p$ in $\mathcal{O}\left(p\log p\right)$ time.

This paper extends this geometric interpretation in order to allow arbitrary convex polygons $P$ with $k$ vertices and perimeter $p$. Also, this extended algorithm only needs $\mathcal{O}\left(k + p(\log p)^2 \log k\right)$ time.

Additionally, this paper presents fast algorithms for counting sub-cadences and cadences with 3 elements using this extended method.
\end{abstract}


\section{Introduction}

The convolution is a well-known and very useful method, which is not only closely linked to signal processing (e.g.\ \cite{ConvAppThompson}) but
is also used to multiply polynomials (see \cite[p.~905]{Cormen}) and large numbers (e.g.\ \cite{Schoenhage} (written in German)) in quasi-linear time. The convolution can be efficiently computed with the fast Fourier transform or its counterpart in residue class rings, the number theoretic transform:

\begin{theorem}
	Let $a=(a_0, a_1, a_2, \dots, a_{n-1})$ and $b=(b_0, b_1, b_2, \dots, b_{n-1})$ be two sequences.
	The sequence $c=(c_0, c_1, c_2, \dots, c_{2n-2})$ with $c_k = \sum_{i+j=k} \left(a_i b_j\right)$ can be computed in $\mathcal{O}\left(n\log n\right)$ operations.
\end{theorem}

The most well-known proofs use additions and multiplications of arbitrary complex numbers. However, with the finite register lengths of real-world computers, one must either cope with the roundoff errors or do all calculations in a different ring. In Appendix \ref{convolutions}, we show that a suitable ring is only dependent on $\left\lfloor \log n\right\rfloor$ and can be found in $\mathcal{O}\left(n (\log n)^2 (\log\log n)\right)$ time if the generalized Riemann hypothesis is true.

The convolution can also be interpreted geometrically: Let $a=(a_0, a_1, a_2, \dots, a_{n-1})$ and $b=(b_0, b_1, b_2, \dots, b_{n-1})$ be sequences. Then the convolution calculates the partial sums
\[\displaystyle\sum_{\substack{i+j=k\\ (i,j)\in P\cap \mathbb{Z}^2}} a_i b_j\textup{,}\]
where $P$ is the square given by $\{(x,y): 0\leq x,y \leq n-1 \}$.

This paper extends this geometric interpretation and shows that if $P$ is an arbitrary convex polygon with $k$ vertices and perimeter $p$, the partial sums can be calculated in $\mathcal{O}\left(k+p(\log p)^2 \log k\right)$ time.

We also use this extended method to solve an open problem of a string pattern called cadence. A cadence is given by an arithmetic progression of occurrences of the same character in a string such that the progression can not be extended to either side without extending the string as well. For example, in the string $001001001$ the indices $(3,6,9)$ corresponding to the ``1''s form a $3$-cadence. On the other hand, in the string $001010100$ the indices $(3,5,7)$ corresponding to the ``1''s do not form a $3$-cadence since, for example, the index $1$ is still inside of the string.

$3$-cadences can be found na\"ively in quadratic time. In the paper \cite{Cadences_Amir}, a quasi-linear time algorithm for detecting the existence of $3$-cadences was proposed, but this algorithm also detects false positives as the aforementioned string $001010100$.

This paper fixes this issue and also extends the algorithm to the slightly more general notion of $(a,b,c)$-partial-$k$-cadences. The resulting extended algorithm also allows counting those partial-cadences and only needs $\mathcal{O}\left(n (\log n)^2\right)$ time.
Using a method presented by Amir et al.\ in \cite{Cadences_Amir}, this implies that all $(a,b,c)$-partial-$k$-cadences can be counted in $\mathcal{O}\left(\min(|\Sigma|n(\log n)^2, n^{3/2}\log n)\right)$ time.

Furthermore, we show that the output of the counting algorithm also allows for finding $o$ partial-cadences in $\mathcal{O}\left(on\right)$ time.

This paper also gives similar results for $3$-sub-cadences.

For the time complexity, we assume that arithmetic operations with $\mathcal{O}(\log n )$ bits can be done in constant time. In particular, we want to be able to get the remainder of a division by a prime $p<2(2n\log(2n))^2$ in constant time.

Also, in this paper, we assume a suitable alphabet. I.e.\ the characters are given by sufficiently small integers in order to allow constant time reading of a given character in the string and in order to allow sorting the characters.

\section{(Sub-)Cadences and Their Definitions}

The term cadence in the context of strings dates back to 1964 and was first introduced by Gardelle and Guilbaud in \cite{Cadences_Gardelle} (written in French). Since then, there were at least two other, slightly different and non-equivalent definitions given by Lothaire in \cite{Cadences_Lothaire} and Amir et al.\ in \cite{Cadences_Amir}.

This paper uses the most restrictive definition of the cadence, which was introduced by Amir et al.\ in \cite{Cadences_Amir}, and also uses their definition of the sub-cadence, which is equivalent to Gardelle's cadence in \cite{Cadences_Gardelle} and Lothaire's arithmetic cadence in \cite{Cadences_Lothaire}.

A string $S$ of length $n$ is the concatenation $S=S[1..n]=S[1]S[2]S[3]\dots S[n]$ of characters from an alphabet $\Sigma$.

\begin{definition}
	A \emph{$k$-sub-cadence} is a triple $(i,d,k)$ of positive integers such that
	\[S[i] = S[i+d] = S[i+2d] = \dots = S[i+(k-1)d]\]
	holds.
\end{definition}

In this paper, cadences are additionally required to start and end close to the boundaries of the string:

\begin{definition}
	A \emph{$k$-cadence} is a $k$-sub-cadence $(i,d,k)$ such that the inequalities
	$i-d \leq 0$ and $n < i+kd$
	hold.
\end{definition}

Since for any $k$-sub-cadence the inequality $i+(k-1)d \leq n$ holds, for any $k$-cadence $i+(k-1)d \leq n < i+kd$ holds.
This implies $k-1 \leq \frac{n-i}{d} < k$ and thereby $k = \left\lfloor \frac{n-i}{d} \right\rfloor + 1$.
It is therefore sufficient to omit the variable $k$ of a $k$-cadence $(i,d,k)$ and just denote this $k$-cadence by the pair $(i,d)$.

\begin{remark}[Comparison of the Definitions]\
	\begin{itemize}
		\item The cadence as defined by Lothaire is just an ordered sequence of unequal indices such that the corresponding characters are equal.
		\item The cadence as defined by Gardelle and Guilbaud additionally requires the sequence to be an arithmetic sequence.
		\item The cadence as defined by Amir et al.\ and as used in this paper additionally requires that the cadence can not be extended in any direction without extending the string as well.
	\end{itemize}
\end{remark}

For the analysis of cadences with errors, we need two more definitions:

\begin{definition}
	A \emph{$k$-cadence with at most $m$ errors} is a tuple $(i,d,k,m)$ of integers such that $i,d,k\geq 1$ and $i-d\leq 0$ and $n<i+kd$ hold and such that there are $k-m$ different integers $\pi_j\in\{0,1,2,\dots,k-1\}$ with $j = 1, 2, 3, \dots, k-m$ and
	\[S[i+\pi_1 d] = S[i+\pi_2 d] = S[i+\pi_3 d] \dots = S[i+\pi_{k-m} d]\textup{.}\]
\end{definition}

A particularly interesting case of cadences with errors is given by the partial-cadences in which we know all positions where an error is allowed:

\begin{definition}
	For some different integers $\pi_j\in\{0,1,2,\dots,k-1\}$ with $j= 1, 2, 3, \dots, p$, a \emph{$(\pi_1, \pi_2, \pi_3, \dots, \pi_p)$-partial-$k$-cadence} is a triple $(i,d,k)$ of positive integers with $i-d\leq 0$ and $n<i+kd$ such that
	\[S[i+\pi_1 d] = S[i+\pi_2 d] = S[i+\pi_3 d] \dots = S[i+\pi_{p} d]\]
	hold.
\end{definition}

\section{3-Sub-Cadences and Rectangular Convolutions}

Lothaire showed over 20 years ago that sufficiently large strings are guaranteed to have sub-cadences of a given length:

\begin{theorem}[Existence of Sub-Cadences (Lothaire \cite{Cadences_Lothaire})]\

	Let $\Sigma$ be an alphabet and $k$ an integer.
	There exists an integer $N = N(|\Sigma|, k)$ such that every string containing at least $N$ characters has at least one $k$-sub-cadence
\end{theorem}

However, this theorem does not provide the number of $k$-sub-cadences of a given string.

In this section, we will show that $3$-sub-cadences with a given character of a string of length $n$ can be efficiently counted in $\mathcal{O}\left(n\log n \right)$ time.
We will also show that arbitrary $3$-sub-cadences of a string of length $n$ can be counted in $\mathcal{O}\left(n^{3/2}(\log n )^{1/2}\right)$ time and that both counting algorithms allow to output $o$ different $3$-sub-cadences in $\mathcal{O}\left(on\right)$ additional time if at least $o$ different $3$-sub-cadences exist.

Let $\sigma\in\Sigma$ be a character. We will now count all $3$-sub-cadences with character $\sigma$.

Let $(i,d)$ be a $3$-sub-cadence. Since $i+d=\frac{i+(i+2d)}{2}$ holds, the position $i+d$ of the middle occurrence of $\sigma$ only depends on the sum of the index $i$ of first occurrence and the index $i+2d$ of the third occurrence but does not depend on the individual indices of those two positions. Therefore, it is possible to determine the candidates for the middle occurrences with the convolution of the candidates of the first occurrence and the candidates of the third occurrence.

Let the sequence $\delta=(\delta_0, \delta_1, \delta_2, \dots, \delta_n)$ be given by the indicator function for $\sigma$ in $S$:
\[\delta_i := \begin{cases} 1 &\textup{if $S[i]=\sigma$} \\ 0 &\textup{if $S[i]\neq\sigma$ (this includes $i=0$)}\end{cases}\]

With this definition, the product $\delta_i \delta_j$ is $1$ if and only if $S[i]=S[j]=\sigma$ and otherwise is $0$. Therefore $c_k = \sum_{i+j=k} \left(\delta_i \delta_j\right) = \#\{i: S[i]=S[k-i]=\sigma\}$ counts in how many ways the index $\frac{k}{2}$ lies in the middle of two $\sigma$. These partial sums can be calculated in $\mathcal{O}\left(n \log n\right)$ time by convolution.

If $k$ is odd or $S\left[\frac{k}{2}\right]\neq\sigma$ holds, the index $\frac{k}{2}$ can not be the middle index of a $3$-sub-cadence.
If $S\left[\frac{k}{2}\right]=\sigma$ holds, the indicator function $\delta_{\frac{k}{2}}$ is $1$, and therefore $\delta_{\frac{k}{2}}\delta_{\frac{k}{2}}=1$ holds as well. Since $(\delta_{\frac{k}{2}},0,3)$ is not a $3$-sub-cadence, the output element $c_k$ contains one false positive. Additionally, for $i+j=k$ with $i\neq j$ and $S[i]=S[j]=\sigma$, the output element counts the combination $\delta_i \delta_j$ as well as $\delta_j \delta_i$.

Therefore,
\[s_k := \begin{cases} \frac{c_{2k}-1}{2} &\textup{if $S[k]=\sigma$}\\ 0 &\textup{if $S[k]\neq \sigma$}\end{cases}\]
counts exactly the number of $3$-sub-cadences with character $\sigma$ such that the second occurrence of $\sigma$ has index $k$. The sum of the $s_k$ is the number of total $3$-sub-cadences with character $\sigma$.

Also, for each $s_k\neq 0$, all those $s_k$ $3$-sub-cadences can be found in $\mathcal{O}(k)\subseteq\mathcal{O}(n)$ time by checking for each index $i<k$ whether $S[i]=S[k]=S[2k-i]=\sigma$ holds.

If the character $\sigma$ is rare, we can also follow the idea of Amir et al.\ in \cite{Cadences_Amir} for detecting $3$-cadences with rare characters: If all $n_\sigma$ occurrences of the character are known, the $c_k$ can be computed in $\mathcal{O}(n_\sigma^2)$ time by computing every pair of those occurrences. Therefore:

\begin{theorem}
	For every character $\sigma\in\Sigma$, the $3$-sub-cadences with $\sigma$ can be counted in $\mathcal{O}(n\log n)$ time. Also, if all $n_\sigma$ occurrences of $\sigma$ are known, the $3$-sub-cadences with $\sigma$ can be counted in $\mathcal{O}(n_\sigma^2)$ time.
\end{theorem}

Following the proof in \cite{Cadences_Amir}, we can get all occurrences of every character by sorting the input string in $\mathcal{O}\left(n \log n\right)$ time. This implies that the algorithm needs at most\\ $\mathcal{O}\left(\sum_{\sigma\in\Sigma} \min(n_\sigma^2, n\log n)\right)
\subseteq\mathcal{O}\left(\frac{n}{(n\log n)^{1/2}} n\log n\right) = \mathcal{O}(n^{3/2}(\log n)^{1/2})$ time.

\begin{theorem}
	The number of all $3$-sub-cadences can be counted in
	\[\mathcal{O}\left(\min(|\Sigma|n\log n, n^{3/2}(\log n)^{1/2})\right) \textup{ time.}\]
\end{theorem}

\begin{theorem}
	After counting at least $o$ $3$-sub-cadences, it is possible to output $o$ $3$-sub-cadences in $\mathcal{O}(on)$ time.
\end{theorem}

\section{Non-Rectangular Convolutions}

In this section, we will extend the geometric interpretation of the convolution and show that for convex polygons $P$ with $k$ vertices and perimeter $p$ it is possible to calculate the partial sums
\[\displaystyle c_k = \sum_{\substack{i+j=k\\ (i,j)\in P\cap \mathbb{Z}^2}} a_i b_j\]
in $\mathcal{O}\left(k+p (\log p)^2 \log k\right)$ time.

Let's imagine a graph where all integer-coordinates $(i,j)$ have the value $f(i,j) := a_i b_j$.
We don't need the convolution in order to determine the sum of the function values in a given rectangle since we can use the simple factorization $\sum_{i=0}^{n} \sum_{j=0}^{m} \left(a_i b_j\right) = \big(\sum_{i=0}^{n} a_i\big) \big(\sum_{j=0}^{m} b_j\big)$ in $\mathcal{O}(n+m)$ time.
However, the convolution provides the $2n$ partial sums on the $45^{\circ}$-diagonals in almost the same time of $\mathcal{O}\left((n+m)\log (n+m)\right)$.

We will now extend this geometric interpretation firstly to triangles with a vertical cathetus and a horizontal cathetus, then to arbitrary triangles and lastly to convex polygons. In order to do this, we will divide the given polygon $P$ in polygons $P^+_p$ and $P^-_m$ such that for each integer point $(i,j)$ the equality
\[\#\{P^+_p| (i,j)\in P^+_p\}-\#\{P^-_m| (i,j)\in P^-_m\} = \begin{cases} 1 &\textup{if $(i,j)\in P$} \\ 0 &\textup{if $(i,j)\notin P$}\end{cases}\]
holds, and we define
\[(c_p)_k := \sum_{\substack{i+j=k\\ (i,j)\in P^+_p\cap \mathbb{Z}^2}} a_i b_j \textup{ and }(c_m)_k := -\sum_{\substack{i+j=k\\ (i,j)\in P^-_m\cap \mathbb{Z}^2}} a_i b_j\textup{.}\]
By construction, $c_k = \left(\sum (c_p)_k\right) + \left(\sum (c_m)_k\right)$ holds. However, if the edges and vertices of the polygons $P^+_p$ and $P^-_m$ contain integer-points, we need to carefully decide for every of these polygons, which edges and vertices are supposed to be included in the polygons and which are excluded from the polygons.

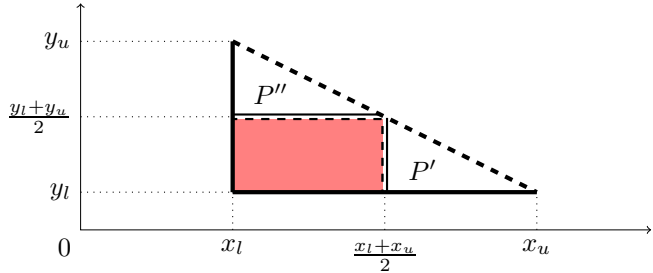
\begin{figure}[!ht]
	\centering
	\begin{tikzpicture}[scale=1]
	\draw [<->, very thin] (0,3) |- (7.5,0);

	\draw (0,0) node [below left] {$0$};
	\draw (0,0.5) node [left] {$y_{l}$};
	\draw (0,1.5) node [left] {$\frac{y_{l}+y_{u}}{2}$};
	\draw (0,2.5) node [left] {$y_{u}$};
	\draw (2,0) node [below] {$x_{l}$};
	\draw (4,0) node [below] {$\frac{x_{l}+x_{u}}{2}$};
	\draw (6,0) node [below] {$x_{u}$};
	\draw (2.5,1.8) node {$P''$};
	\draw (4.5,0.8) node {$P'$};

	\draw[draw=black, fill=red, dashed, thick, opacity=0.5] (2,0.5) rectangle (3.97,1.47);
    \draw[draw=black, thick, dashed] (2,0.5) rectangle (3.97,1.47);

	\draw [ultra thick] (2,0.5) -- (2,2.5);
	\draw [ultra thick, dashed] (2,2.5) -- (6,0.5);
	\draw [ultra thick] (6,0.5) -- (2,0.5);
	\draw [thick] (2,1.53) -- (3.94,1.53);
	\draw [thick] (4.03,1.485) -- (4.03,0.5);

	\draw[dotted] (0,0.5) -- (2,0.5);
	\draw[dotted] (0,1.5) -- (2,1.5);
	\draw[dotted] (0,2.5) -- (2,2.5);
	\draw[dotted] (2,0) -- (2,0.5);
	\draw[dotted] (4,0) -- (4,0.5);
	\draw[dotted] (6,0) -- (6,0.5);
	\end{tikzpicture}
	\caption{The right-angled triangle $P$ in Lemma \ref{lem:simpleTriangleAlg}.}
	\label{fig:simpleTriangleAlg}
\end{figure}

\begin{lemma} \label{lem:simpleTriangleAlg}
    Let $P$ be a triangle with a vertical cathetus and a horizontal cathetus and perimeter $p$. Let also the sequences $a=(a_0, a_1, a_2, \dots, a_n)$ and $b=(b_0, b_1, b_2, \dots, b_n)$ be given.

    Then the partial sums
    \[\displaystyle c_k = \sum_{\substack{i+j=k\\ (i,j)\in P\cap \mathbb{Z}^2}} a_i b_j\]
    can be calculated in $\mathcal{O}\left(p (\log p)^2\right)$ time.
\end{lemma}
\begin{proof}
    The proof will be symmetrical with regard to horizontal and vertical mirroring. Therefore, without loss of generality, we will assume that $P$ is oriented as in Figure \ref{fig:simpleTriangleAlg}.

    We first initialize the output vector $c=(c_{x_{l}+y_{l}}, c_{x_{l}+y_{l}+1}, c_{x_{l}+y_{l}+2}, \dots, c_{x_{u}+y_{u}})$ with zero. This takes $\mathcal{O}\left(p\right)$ time.

    In the following proof, we assume that both catheti are included in the polygon and that the hypotenuse as well as its endpoints are excluded. If this is not the expected behavior, we can traverse the edges in $\mathcal{O}\left(p\right)$ time and for each integer-point $(i,j)$ on the edge, we can decrease/increase the corresponding $c_{i+j}$ by $a_i b_j$ if necessary.

    If $p$ is at most one, there is at most one integer-point $(i,j)$ in the triangle, and this point can be found in constant time. In this case, we only have to increase $c_{i+j}$ by $a_i b_j$.

    If $p$ is bigger than one, we will separate the triangle $P$ into three disjoint parts as seen in Figure \ref{fig:simpleTriangleAlg}.
    \begin{itemize}
        \item The triangle $P'$ of points with x-coordinate of at least $\left\lceil\frac{x_{l}+x_{u}}{2}\right\rceil$,
        \item the triangle $P''$ of points with y-coordinate of at least $\left\lceil\frac{y_{l}+y_{u}}{2}\right\rceil$ and
        \item the red rectangle of points with x-coordinate of at most $\left\lceil\frac{x_{l}+x_{u}}{2}\right\rceil-1$ and y-coordinate of at most $\left\lceil\frac{y_{l}+y_{u}}{2}\right\rceil-1$.
    \end{itemize}

    There are no integers bigger than $\left\lceil\frac{x_{l}+x_{u}}{2}\right\rceil-1$ but smaller than $\left\lceil\frac{x_{l}+x_{u}}{2}\right\rceil$ nor integers bigger than $\left\lceil\frac{y_{l}+y_{u}}{2}\right\rceil-1$ but smaller than $\left\lceil\frac{y_{l}+y_{u}}{2}\right\rceil-1$. Therefore, each integer-point in $P$ is in exactly one of the three parts.

    For the red rectangle, we can calculate the convolution and thereby get the corresponding partial sums in $\mathcal{O}\left(p\log p\right)$ time. The partial sums corresponding to the sub-triangles are calculated recursively. Increasing the $c_k$ by the partial results leads to the final result.

    Hence, the algorithm takes \[\mathcal{O}\left(p+\left(\sum_{i=0}^{\log_2 p} 2^i\left(\frac{p}{2^i} \log \frac{p}{2^i}\right)\right) + 2^{\log_2 p}\right) \subseteq \mathcal{O}\left(\sum_{i=0}^{\log p} p \log p\right) = \mathcal{O}\left(p (\log p)^2\right)\] time.
\end{proof}

We will now further extend this result to arbitrary triangles:

\begin{figure}[!ht]
	\centering
	\begin{tikzpicture}[scale=1]
    \path[fill=blue, opacity=0.5] (1.5,0.5) -- (4.5,0.5) -- (4.5,2) -- (1.5,2) -- cycle;
    \path[fill=red, opacity=0.5] (0.5,0.5) -- (1.5,0.5) -- (1.5,2) -- cycle;
    \path[fill=red, opacity=0.5] (4.5,2) -- (4.5,2.5) -- (1.5,2) -- cycle;
    \path[fill=white, opacity=0.5] (0.5,0.5) -- (4.5,0.5) -- (4.5,2.5) -- cycle;
    \draw[thin] (1.5,0.5) -- (1.5,2) -- (4.5,2);

    \path[fill=red, opacity=0.5] (11.5,2.5) -- (11.5,1.5) -- (8.5,2.5) -- cycle;
    \path[fill=red, opacity=0.5] (7.5,0.5) -- (7.5,2.5) -- (8.5,2.5) -- cycle;
    \path[fill=red, opacity=0.5] (7.5,0.5) -- (11.5,1.5) -- (11.5,0.5) -- cycle;

    \draw [<->, very thin] (0,3) |- (5,0);
    \draw [<->, very thin] (7,3) |- (12,0);

    \draw (0,0) node [below left] {$0$};
    \draw (7,0) node [below left] {$0$};
    \draw (0,0.5) node [left] {$y_{l}$};
    \draw (7,0.5) node [left] {$y_{l}$};
    \draw (0,2.5) node [left] {$y_{u}$};
    \draw (7,2.5) node [left] {$y_{u}$};

    \draw (0.5,0) node [below] {$x_{l}$};
    \draw (7.5,0) node [below] {$x_{l}$};
    \draw (4.5,0) node [below] {$x_{u}$};
    \draw (11.5,0) node [below] {$x_{u}$};

    \draw[draw=black, dashed, thin] (0.5,0.5) rectangle (4.5,2.5);
    \draw[draw=black, dashed, thin] (7.5,0.5) rectangle (11.5,2.5);

    \draw [ultra thick, dashed] (0.5,0.5) -- (4.5,2.5);
    \draw [ultra thick, dashed] (0.5,0.5) -- (1.5,2);
    \draw [ultra thick, dashed] (1.5,2) -- (4.5,2.5);

    \draw [ultra thick, dashed] (7.5,0.5) -- (11.5,1.5);
    \draw [ultra thick, dashed] (7.5,0.5) -- (8.5,2.5);
    \draw [ultra thick, dashed] (8.5,2.5) -- (11.5,1.5);
	\end{tikzpicture}
	\caption{The two possible triangles $P$ in Lemma \ref{lem:difficultTriangleAlg}.}
	\label{fig:difficultTriangleAlg}
\end{figure}
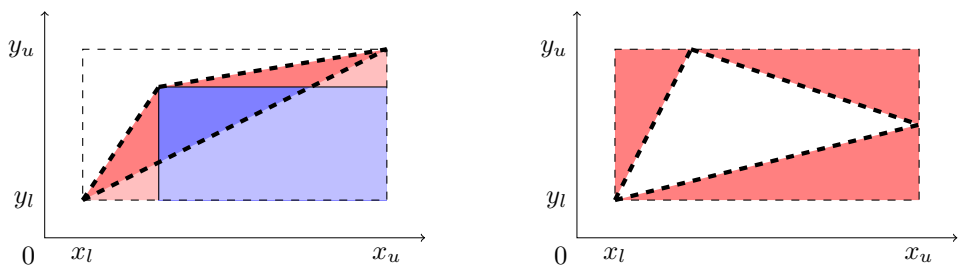

\begin{lemma} \label{lem:difficultTriangleAlg}
    Let a triangle $P$ with perimeter $p$ and sequences $a=(a_0, a_1, a_2, \dots, a_n)$ and $b=(b_0, b_1, b_2, \dots, b_n)$ be given.

    Then the partial sums
    \[\displaystyle c_k = \sum_{\substack{i+j=k\\ (i,j)\in P\cap \mathbb{Z}^2}} a_i b_j\]
    can be calculated in $\mathcal{O}\left(p (\log p)^2\right)$ time.
\end{lemma}
\begin{proof}
    Let $x_l, y_l, x_u, y_u$ be the minimal and maximal x-coordinates and y-coordinates of the three vertices of the polygon $P$.
    As in the last lemma, we first initialize the output vector $c = (c_{x_l+y_l}, c_{x_l+y_l+1}, c_{x_l+y_l+2}, \dots, c_{x_u+y_u})$.

    Similarly to the last lemma, we can remove/add edges and vertices in linear time with respect to $p$. Since the number of edges and vertices is constant, we ignore them for the sake of simplicity.

    Let $R$ be the rectangle $\{(x,y)| x_l<x<x_u \land y_l<y<y_u\}$. Since $R$ has four edges but $P$ only has three vertices, at least one of the vertices of $P$ is also a vertex of $R$. Without loss of generality, this vertex is $(x_l, y_l)$.

    \begin{description}
        \item[Case 1:] The opposing vertex $(x_u,y_u)$ in $R$ also coincides with a vertex of $P$ (as in the left hand side of Figure \ref{fig:difficultTriangleAlg}):

        Without loss of generality, we can assume that the third vertex of $P$ is above the diagonal from $(x_l,y_l)$ to $(x_u,y_u)$. In this case, the partial sums corresponding to $P$ are given by the sum of the partial sums of the red triangles and the partial sums of the blue rectangle minus the partial sums of the lighter triangle.

        There are only three triangles and one rectangle involved, and each of those polygons has perimeter $\mathcal{O}\left(p\right)$. Furthermore, all triangles have a vertical cathetus and a horizontal cathetus. Therefore, using Lemma \ref{lem:simpleTriangleAlg}, we can calculate all partial sums in $\mathcal{O}\left(p (\log p)^2\right)$ time.

        \item[Case 2:] The opposing vertex $(x_u,y_u)$ in $R$ does not coincide with a vertex of $P$ (as in the right hand side of Figure \ref{fig:difficultTriangleAlg}):

        In this case, one vertex of $P$ lies on the right edge of $R$ and one vertex of $P$ lies on the upper edge of $R$.

        The wanted partial sums are in this case the difference of the partial sums of the rectangle and of the partial sums of the three red triangles. Again, we can calculate all partial sums in $\mathcal{O}\left(p (\log p)^2\right)$ time.
    \end{description}

    Since both cases require $\mathcal{O}\left(p (\log p)^2\right)$ time, this concludes the proof.
\end{proof}

Now we will extend this algorithm to convex polygons by dissecting them into triangles with sufficiently small perimeter.

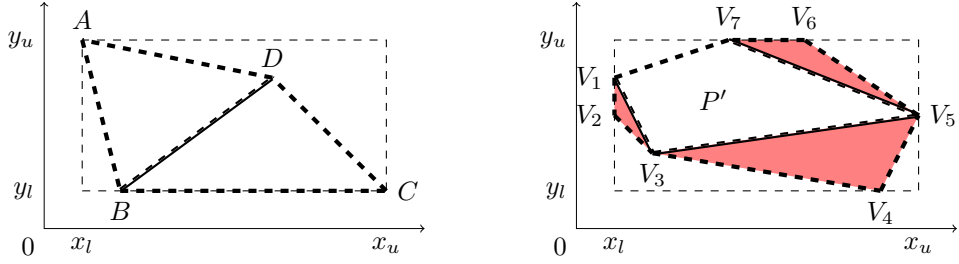
\begin{figure}[!ht]
    \centering
    \begin{tikzpicture}[scale=1]
    \path[fill=red, opacity=0.5] (7.5,2) -- (7.5,1.5) -- (8,1) -- cycle;
    \path[fill=red, opacity=0.5] (8,1) -- (11,0.5) -- (11.5,1.5) -- cycle;
    \path[fill=red, opacity=0.5] (11.5,1.5) -- (10,2.5) -- (9,2.5) -- cycle;

    \draw [<->, very thin] (0,3) |- (5,0);
    \draw [<->, very thin] (7,3) |- (12,0);

    \draw (0,0) node [below left] {$0$};
    \draw (7,0) node [below left] {$0$};
    \draw (0,0.5) node [left] {$y_{l}$};
    \draw (7,0.5) node [left] {$y_{l}$};
    \draw (0,2.5) node [left] {$y_{u}$};
    \draw (7,2.5) node [left] {$y_{u}$};

    \draw (0.5,0) node [below] {$x_{l}$};
    \draw (7.5,0) node [below] {$x_{l}$};
    \draw (4.5,0) node [below] {$x_{u}$};
    \draw (11.5,0) node [below] {$x_{u}$};

    \draw[draw=black, dashed, thin] (0.5,0.5) rectangle (4.5,2.5);
    \draw[draw=black, dashed, thin] (7.5,0.5) rectangle (11.5,2.5);

    \draw [ultra thick, dashed] (1,0.5) node [below] {$B$} -- (4.5,0.5);
    \draw [ultra thick, dashed] (4.5,0.5) node [right] {$C$} -- (3,2);
    \draw [ultra thick, dashed] (3,2) node [above] {$D$} -- (0.5,2.5);
    \draw [ultra thick, dashed] (0.5,2.5) node [above] {$A$} -- (1,0.5);

    \draw [thick] (1.013,0.5) -- (3.005,1.980);
    \draw [thick, dashed] (0.996,0.524) -- (2.988,2.005);

    \draw [ultra thick, dashed] (7.5,2) node [left] {$V_1$} -- (7.5,1.5);
    \draw [ultra thick, dashed] (7.5,1.5) node [left] {$V_2$} -- (8,1);
    \draw [ultra thick, dashed] (8,1) node [below] {$V_3$} -- (11,0.5);
    \draw [ultra thick, dashed] (11,0.5) node [below] {$V_4$} -- (11.5,1.5);
    \draw [ultra thick, dashed] (11.5,1.5) node [right] {$V_5$} -- (10,2.5);
    \draw [ultra thick, dashed] (10,2.5) node [above] {$V_6$} -- (9,2.5);
    \draw [ultra thick, dashed] (9,2.5) node [above] {$V_7$} -- (7.5,2);

    \draw [thick] (7.5,1.983) -- (7.987,1.005);
    \draw [thick, dashed] (7.525,2) -- (8.023,1);

    \draw [thick] (8.01,0.985) -- (11.51,1.485);
    \draw [thick, dashed] (8.01,1.015) -- (11.51,1.515);

    \draw [thick] (11.51,1.514) -- (9.01,2.5134);
    \draw [thick, dashed] (11.51,1.482) -- (9.01,2.482);

    \draw (8.8,1.7) node {$P'$};

    \end{tikzpicture}
    \caption{Two possible convex polygons $P$ with more than $3$ vertices in Lemma \ref{lem:convexAlg}.}
    \label{fig:convexAlg}
\end{figure}

\begin{theorem} \label{lem:convexAlg}
    Let $P$ be a convex polygon with $k$ vertices and perimeter $p$. Let also the sequences $a=(a_0, a_1, a_2, \dots, a_n)$ and $b=(b_0, b_1, b_2, \dots, b_n)$ be given.

    Then the partial sums
    \[\displaystyle c_k = \sum_{\substack{i+j=k\\ (i,j)\in P\cap \mathbb{Z}^2}} a_i b_j\]
    can be calculated in $\mathcal{O}\left(k + p (\log p)^2 \log k\right)$ time.
\end{theorem}
\begin{proof}
    As in the last two Lemmata, we define $x_l, y_l, x_u, y_u$ to be the minimal and maximal x-coordinates and y-coordinates of the $k$ vertices of $P$. Also, we first initialize the output vector $c = (c_{x_l+y_l}, c_{x_l+y_l+1}, c_{x_l+y_l+2}, \dots, c_{x_u+y_u})$. We further assume that none of the edges and vertices of $P$ is included in $P$.

    If $P$ is a triangle, then this Lemma simplifies to Lemma \ref{lem:difficultTriangleAlg} and there is nothing left to prove.

    If $P$ is a quadrilateral $ABCD$, as in the left hand side of Figure \ref{fig:convexAlg}, then it can be partitioned into the triangles $ABD$ and $CDB$ where the edge $BD$ is included in exactly one triangle and all other edges are excluded. The triangle inequality proves that $|BD| \leq |DA|+|AB|$ and $|BD| \leq |BC|+|CD|$ hold. Therefore, both triangles have a perimeter of at most $p$. This implies that the partial sums can be calculated in $\mathcal{O}\left(p (\log p)^2\right)$

    If $P$ is a polygon $V_1 V_2 V_3 \dots V_k$ with more than four vertices, as in the right hand side of Figure \ref{fig:convexAlg}, it can be partitioned into
    \begin{itemize}
    	\item the polygon $P' = V_1 V_3 V_5 \dots V_{2\left\lceil\frac{k}{2}\right\rceil-1}$, which is given by the odd vertices without its edges,
    	\item the red triangles $V_i V_{i+1} V_{i+2}$ with $i = 1, 3, 5, \dots, 2\left\lceil\frac{k}{2}\right\rceil-3$ including the edge $V_i V_{i+2}$ but excluding the other edges and the vertices,
    	\item if $k$ is even, the triangle $V_{k-1} V_{k}$ including the edge $V_{k-1} V_{k+1}$ but excluding the other edges and the vertices.
    \end{itemize}
	By construction and triangle inequality, the perimeter $p'$ of $P'$ is at most $p$. This, however, also implies that the total perimeter $\sum p_i$ of the triangles is at most $2p$. The inequality
	\[\sum \min\left(1, p_i (\log p_i)^2\right) \leq k + \sum \left(p_i (\log p)^2\right) \leq k + p(\log p)^2\]
	implies that the algorithm needs $\mathcal{O}\left(k + p(\log p)^2\right)$ time plus the time we need for processing $P'$. Since each step almost halves the number of vertices, we need $\mathcal{O}\left(\log k\right)$ steps. This results in a total time complexity of $\mathcal{O}\left(k + p (\log p)^2 \log k\right)$.
\end{proof}

\section{(a,b,c)-Partial-k-Cadences}

In this section, we will show how the non-rectangular convolution helps counting the $(a,b,c)$-partial-$k$-cadences with a given character $\sigma$ in $\mathcal{O}\left(n (\log n)^2\right)$. We will further show that all $(a,b,c)$-partial-$k$-cadences can be counted in $\mathcal{O}\left(\min(|\Sigma|n(\log n)^2, n^{3/2}\log n)\right)$ time and that both counting algorithms allow to output $o$ of those partial-cadences in $\mathcal{O}\left(on\right)$ time.

As a special case, these results also hold for $3$-cadences.

We further conclude from these results that the existence of $k$-cadences with at most $k-3$ errors can be detected in $\mathcal{O}\left(\min(|\Sigma| k^3 n(\log n)^2, k^3 n^{3/2}\log n)\right)$ time.

Without loss of generality, we will only deal with the case $a<b$ in this section.

\begin{lemma}\label{lem:boundaries}
    Three positions $x$, $y$ and $z$ form a $(a,b,c)$-partial-$k$-cadence if and only if
    \begin{itemize}
    \item the equation $\frac{y-x}{b-a} = \frac{z-y}{c-b} \in \mathbb{Z}$ holds,
    \item the equation $S[x] = S[y] = S[z]$ holds and
    \item the inequalities
    \begin{align}
        0&\geq \frac{(b+1)x-(a+1)y}{b-a}\textup{,}\\
        0&<\frac{bx-ay}{b-a}\textup{,}\\
        n&\geq \frac{(b-k+1)x-(a-k+1)y}{b-a}\textup{ and}\\
        n&< i+kd = \frac{(b-k)x-(a-k)y}{b-a}\textup{ hold.}
    \end{align}
    \end{itemize}
\end{lemma}
\begin{proof}
    Define $d:=\frac{y-x}{b-a}$ and $i:=x-ad$. Then $x=i+ad$ and $y=i+bd$.
    Furthermore, the equation $\frac{y-x}{b-a} = \frac{z-y}{c-b}$ holds if and only if $z=i+cd$ and $\frac{y-x}{b-a}\in \mathbb{Z}$ holds if and only if $d$ is an integer.

    Additionally, using $x=i+ad$ and $y=i+bd$, the four inequalities can be simplified to $0\geq i-d$, $0<i$, $n\geq i+(k-1)d$ and $n<i+kd$.

    Therefore, the lemma follows from the definition of the partial-cadence.
\end{proof}

\begin{figure}[!ht]
    \centering
    \begin{tikzpicture}[scale=1]
    \draw [<->, very thin] (0,4.8) |- (4.8,0);

    \draw (4.8,0) node [right] {$x$};
    \draw (0,4.8) node [above] {$y$};
    \draw (0,0) node [below left] {$0$};
    \draw (5/4,0) node [below] {$\frac{1}{4} n$};
    \draw (10/4,0) node [below] {$\frac{2}{4} n$};
    \draw (0,10/4) node [left] {$\frac{2}{4} n$};
    \draw (0,15/4) node [left] {$\frac{3}{4} n$};

    \draw[thin] (0,0) -- (10/3,5);
    \draw[thin, dashed] (0,0) -- (2.5,5);
    \draw[thin] (0,2.5) -- (5,5);
    \draw[thin, dashed] (0,5/3) -- (5,5);

    \draw (10/3,5) node [right] {(1)};
    \draw (2.5,5) node [left] {(2)};
    \draw (0,2.6) node [above right] {(3)};
    \draw (-0.1,1.9) node [below right] {(4)};
    \draw (5,5) node [right] {$(n,n)$};

    \draw[very thick, dashed] (5/4,10/4) -- (2,3);
    \draw[very thick] (2,3) -- (10/4,15/4);
    \draw[very thick] (10/4,15/4) -- (5/3,10/3);
    \draw[very thick, dashed] (5/3,10/3) -- (5/4,10/4);

    \draw (5/4,10/4) node [below] {A};
    \draw (2,3) node [below right] {B};
    \draw (10/4,15/4) node [right] {C};
    \draw (5/3,10/3) node [above left] {D};

    \draw[dotted] (5/4,0) -- (5/4,10/4);
    \draw[dotted] (0,10/4) -- (5/4,10/4);
    \draw[dotted] (10/4,0) -- (10/4,15/4);
    \draw[dotted] (0,15/4) -- (10/4,15/4);
    \end{tikzpicture}
    \caption{The four inequalities of Lemma \ref{lem:boundaries} for $(1,2,3)$-partial-$4$-cadences.}
    \label{fig:quadrilateral}
\end{figure}
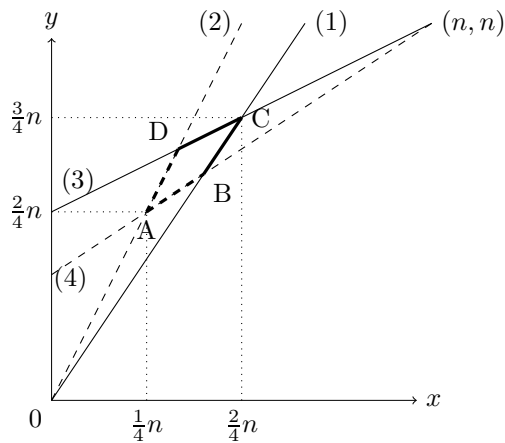

The four inequalities hold if the points $(x,y)$ lie inside the convex quadrilateral given, as shown in Figure \ref{fig:quadrilateral}, by the corners
\begin{align*}
A&=\left(\frac{an}{k},\frac{bn}{k}\right)\\
B&=\left(\frac{(a+1)n}{k+1},\frac{(b+1)n}{k+1}\right)\\
C&=\left(\frac{(a+1)n}{k},\frac{(b+1)n}{k}\right)\\
D&=\left(\frac{an}{k-1},\frac{bn}{k-1}\right)\\
\end{align*}
including the vertex $C$ and the edges between $B$ and $C$ as well as between $C$ and $D$ but excluding all other vertices and the edges between $A$ and $B$ as well as between $D$ and $A$.

For given $x=i+ad$ and $y=i+bd$, the third occurrence $z=i+cd$ can be calculated with the equation $i+cd = \frac{(b-c)(i+ad)+(c-a)(i+bd)}{b-a}$ directly without calculating $i$ and $d$ first. The corresponding partial sums
\[\displaystyle c_k = \sum_{\substack{i+j=k\\ (i,j)\in P\cap \mathbb{Z}^2}} a_{\frac {i} {(b-c)}} b_{\frac{j}{(c-a)}}\]
can be calculated by using the partial sums
\[\displaystyle c_k = \sum_{\substack{i+j=k\\ (i,j)\in P'\cap \mathbb{Z}^2}} a'_{i} b'_{j}\]
with $a'_i := \begin{cases} a_{\frac{i}{b-c}} &\textup{if $i \equiv 0 \pmod{b-c}$} \\ 0 &\textup{otherwise}\end{cases}\textup{ and } b'_j := \begin{cases} b_{\frac{j}{c-a}} &\textup{if $j \equiv 0 \pmod{c-a}$} \\ 0 &\textup{otherwise}\end{cases}$ and a polygon $P'$, which is derived from $P$ by stretching the first coordinate by $(b-c)$ and the second coordinate by $(c-a)$. The perimeter of $P'$ is at most $\max(|b-c|,|c-a|)$ times the perimeter of $P$.
Using the quadrilateral $P=ABCD$ with perimeter
\[p \leq 2|C_x-A_x|+2|C_y-A_y| = 2\left(\frac{(a+1)n}{k}-\frac{an}{k}\right) + 2\left(\frac{(b+1)n}{k}-\frac{bn}{k}\right) = \frac{4n}{k}\in \mathcal{O}\left(\frac{n}{k}\right)\textup{,}\]
the polygon $P'$ has perimeter $p'\in\mathcal{O}\left(n\right)$. This proves the following three theorems.

\begin{theorem}
	For every character $\sigma\in\Sigma$, the $(a,b,c)$-partial-$k$-cadences with $\sigma$ can be counted in $\mathcal{O}(n(\log n)^2)$ time. Also, if all $n_\sigma$ occurrences of $\sigma$ are known, the $(a,b,c)$-partial-$k$-cadences with $\sigma$ can be counted in $\mathcal{O}(n_\sigma^2)$ time.
\end{theorem}

\begin{theorem}
	The number of all $(a,b,c)$-partial-$k$-cadences can be counted in
	\[\mathcal{O}\left(\min(|\Sigma|n(\log n)^2, n^{3/2}\log n)\right)\textup{ time.}\]
\end{theorem}

\begin{theorem}
	After counting at least $o$ $(a,b,c)$-partial-$k$-cadences, it is possible to output $o$ $(a,b,c)$-partial-$k$-cadences in $\mathcal{O}(on)$ time.
\end{theorem}

Since every $3$-cadence is an $(0,1,2)$-partial-$3$-cadence, we also obtain the special case:

\begin{corollary}
    For every character $\sigma\in\Sigma$, the $3$-cadences with $\sigma$ can be counted in $\mathcal{O}(n(\log n)^2)$ time. Also, if all $n_\sigma$ occurrences of $\sigma$ are known, the $3$-cadences with $\sigma$ can be counted in $\mathcal{O}(n_\sigma^2)$ time.

    Therefore, the number of all $3$-cadences can be counted in
    \[\mathcal{O}\left(\min(|\Sigma|n(\log n)^2, n^{3/2}\log n)\right)\textup{ time.}\]

    Also, after counting at least $o$ $3$-cadences, it is possible to output $o$ $3$-cadences in $\mathcal{O}(on)$ time.
\end{corollary}

Taking the sum over all possible triples $(a,b,c)$, we can also search for $k$-cadences with at most $k-3$ errors. It can be checked in
\[\mathcal{O}\left(\min(|\Sigma| k^3 n(\log n)^2, k^3 n^{3/2}\log n)\right)\]
time whether the given string has a $k$-cadence with at most $k-3$ errors.
However, since $k$-cadences with less than $k-3$ errors are counted more than once, it seems to be difficult to determine the exact number of $k$-cadences with at most $k-3$ errors.

\section{Conclusion}

This paper extends convolutions to arbitrary convex polygons. One might wonder whether these convolutions could be speed up or be further extended to non-convex polynomials.

Instead of just partitioning the interior of the polygon into triangles, it is also possible to identify polygons by the difference of a slightly bigger but less complex polygon and a triangle. However, if the algorithm presented in this paper is adapted to non-convex polygons, it can generate self-intersecting polygons. While the time-complexity stays the same for these polygons, it becomes hard to ensure that every vertex and every edge of the polygon is counted exactly once.

Another approach is given by Levcopoulos and Lingas in \cite{LevcopoulosAlgo}. This paper shows that any simple polygon can be decomposed into convex components in quasilinear time with only logarithmic blow-up. This paper also shows that if the input polygon is rectilinear, this partition only contains axis-aligned rectangles. Since the convolution handles rectangles quicker and more easily than triangles, this saves a logarithm. However, in general, it is not obvious how to transform arbitrary polygons into equivalent simple rectilinear polygons in quasilinear time without blowing-up the number of vertices too much.

The non-rectangular convolution, unlike the usual convolution, allows to define a dependence between the indices of the convoluted sequences. This dependence is not usable in applications like the multiplication of polynomials, and for many signal processing applications this extended method does not seem to bring any benefits either. However, in order to count the partial-cadences this dependence was essential. The non-rectangular convolution may also have future applications in image processing and convolutional neural networks.

In terms of cadences, this paper presents algorithms to count and find sub-cadences, cadences and partial-cadences with three elements. However, if there are linearly many $c$-positions of $(a,b,c)$-partial-$k$-cadences, the knowledge of those partial-cadences does not lead to a sub-quadratic-time-algorithm for determining the existence $4$-cadences. On the other hand, it is also not shown that this problem needs quadratic time.

Also, the time-complexity $\mathcal{O}\left(on\right)$ for finding $o$ $3$-cadences is quite pessimistic. If there are many $3$-cadences, it is very likely that quite a few of these $3$-cadences share one of their occurrences. These occurrences can be found in $\mathcal{O}(n)$ time. On the other hand, in the string $10^{n-1}1^{2n}$, for example, there are linearly many $3$-cadences but every second occurrence and every third occurrence only occurs in at most one of those $3$-cadences.

\section{acknowledgements}
	The first author discovered an error in the algorithm for determining the existence of 3-cadences in "String cadences" of Amir et al., which led to false-positives. Travis Gagie explained this error to the second author at the CPM-Conference in Pisa. He also claimed that this problem should be solvable. Juliusz Straszyński showed that $3$-sub-cadences beginning and ending in given intervals can efficiently be detected by convolution. Amihood Amir noted that we can also efficiently count these sub-cadences, which allows ``subtractive'' methods as used for arbitrary triangles.

\appendix

\section{appendix}
\subsection{Convolutions} \label{convolutions}

It is well-known that the discrete convolution can be calculated with $\mathcal{O}\left(n\log n\right)$ complex arithmetic operations.
However, if the convolution is calculated with the fast Fourier transform, the finite register lengths introduce roundoff errors. These errors can propagate and accumulate throughout the calculation.

Therefore, in order to calculate the convolution of integer sequences, it seems more convenient to use the number theoretic transform, which is the generalization of the fast Fourier transform from the field of the complex numbers to certain residue class rings.

In this section, we will show that after some precomputation in $\mathcal{O}\left(n(\log n)^2 (\log\log n)\right)$ time it is possible to calculate these convolutions in $\mathcal{O}\left(n\log n\right)$ time.

Agarwal and Burrus show in \cite{AgarwalNTTandFermat} that the cyclic convolution of two integer-vectors of length $n$ can be efficiently computed modulo a prime $p$ if $p-1$ is a multiple of $n$.

Linnik proves in \cite{Linnik} that there are constants $c$ and $L$ such that for each $n$, $r$ with $\gcd(n,r)=1$, there is a prime of the form $mn+r$ with $mn+r < cn^L$.
While Linnik himself did not provide the values of $c$ and $L$, there are some upper bounds:
For example, Xylouris proves in \cite{XylourisOld} that there is a $c$ such that for each $n$, $r$ with $\gcd(n,r)=1$, there is a prime of the form $mn+r$ with $mn+r < cn^{5.18}$.
More explicitly, Bach and Sorenson present in \cite{BachLinnikGRH} that if the generalized Riemann hypothesis holds, for each $n$, $r$ with $\gcd(n,r)=1$, there is a prime of the form $mn+r$ with $mn+r < 2(n\log n)^2$.

As a result, for each $n$, there is a prime $p_n \equiv 1 \pmod{n}$ with $p_n  < 2(n\log n)^2$. This also implies that the length of $p_n$ is at most $4$ times the length of $n$. Therefore, such a prime number $p_n$ is a good modulus for the convolution of length $n$ or any of its divisors. It is left to show that such a prime $p_n$ can be efficiently found.

\begin{theorem}
    Let $n$ be an integer. A prime $p_n \equiv 1 \pmod{n}$ with $p_n  < 2(n\log n)^2$ can be found in $\mathcal{O}(n(\log n)^2 \log\log (n))$ time.
\end{theorem}
\begin{proof}
    The main idea is to use the sieve of Eratosthenes to first find all primes up to $2n\log n$ and then sieve only the numbers up to $2(n\log n)^2$ that are congruent to $1$ modulo $n$ with these primes.

    On the one hand, since $(2n\log n)^2 > 2(n\log n)^2$ holds, all numbers left after the second sieving are primes.
    On the other hand, the result of Bach and Sorenson in \cite{BachLinnikGRH} guarantees that if the generalized Riemann hypothesis holds, there is a prime left.
    Also, by construction, all primes $p_n$ left fulfill this theorem.

    It remains to be shown that this algorithm can be done in $\mathcal{O}(n(\log n)^2 \log\log(n))$ time.

    For the usual sieve of Eratosthenes, one prepares a Boolean array for the first $2n\log n$ numbers. Then, for each number that has not been marked as non-prime, every multiple is marked as non-prime. Afterwards, all non-marked numbers are returned. The majority of the time is spend for the marking. This takes
    \[\mathcal{O}\left(\sum_{\substack{p=2\\ p\textup{ is prime}}}^{2n\log n} \frac{2n\log n}{p} \right) = \mathcal{O}\left(n\log n \sum_{\substack{p=2\\ p\textup{ is prime}}}^{2n\log n} \frac{1}{p} \right) = \mathcal{O}\left(n(\log n) (\log\log n)\right)\]
    time. The last equality is given by Mertens in \cite[p.~46]{Mertens} (written in German) and the inequality $\log\log (2n\log n) < 2 \log\log (n)$.

    For the second part, we have a much larger interval of numbers. However, since we only have to consider the first residue class, only every $n$-th number has to be considered. Therefore we need
    \[\mathcal{O}\left(\sum_{\substack{p=2\\ p\textup{ is prime}}}^{2n\log n} \frac{2(n\log n)^2}{np} \right) = \mathcal{O}\left(n(\log n)^2 \sum_{\substack{p=2\\ p\textup{ is prime}}}^{2n\log n} \frac{1}{p} \right) = \mathcal{O}\left(n(\log n)^2 (\log\log n)\right)\]
    markings. Using the extended Euclidean algorithm, for every prime $p$, we can find the smallest $f$ such that $fp \equiv 1 \pmod{n}$ in $\mathcal{O}\left(\log p\right) \subseteq \mathcal{O}\left(\log n\right)$ time. Summing up over all primes, this takes
    \[\mathcal{O}\left(\sum_{\substack{p=2\\ p\textup{ is prime}}}^{2n\log n} \log n \right) \subseteq \mathcal{O}\left(n(\log n)^2\right)\]
    time.

    This concludes the proof.
\end{proof}

\begin{remark}
    The prime number theorem states that the number $\pi(N)$ of primes smaller than $N$ asymptotically behaves like $\frac{N}{\log N}$.
    Dirichlet's prime number theorem states that for a given $n$ and a sufficiently large $N$, the prime numbers are evenly distributed in all residue classes $mn+r$ with $\gcd(n,r)=1$.

    Therefore, for a given $n$ and sufficiently large $N$, we should expect circa $\frac{N}{\varphi(n)\log N }$ prime numbers of the form $mn+1$ that are smaller than $N$. One might therefore hope that it is possible to guess logarithmically many numbers smaller than $N$ in the right residue class, and then test in $\mathcal{O}\left((\log N)^c \right)$ time whether this number is prime.

    However, the ``sufficient largeness'' of $N$ depends on $n$. Therefore, these theorems do not provide the number of suitable primes smaller than, for example, $2(n\log n)^2$. Also, since the generation of suitable primes can be done in quasilinear time, the randomized shortcut is not necessary.
\end{remark}

It is not only possible to find a suitable modulus for the number theoretic transform, but we can also find a suitable $2^t$-th root:

\begin{theorem}
    Let $p_{2^t}$ be a prime with $p_{2^t} \equiv 1 \pmod{2^t}$ and $p_{2^t} < 2({2^t}\log({2^t}))^2$.

    A $2^t$-th root of unity modulo $p_{2^t}$ can be found in $\mathcal{O}\left((\log p_{2^t})^3\right)$ time.
\end{theorem}
\begin{proof}
    Let $p_{2^t}=1+o2^r$ for an odd number $o$.

    Firstly, we will show that a residue $q^o$ is a $2^r$-th root of unity modulo $p_{2^t}$ if and only if $q$ is a quadratic nonresidue modulo $p_{2^t}$.

    Since $p_{2^t}$ is prime, there is a primitive root $a$ modulo $p_{2^t}$.

    Let $q\equiv a^i$. Then $q^o=a^{io}$ has the order $\frac{o2^r}{\gcd(io,o2^r)}=\frac{2^r}{\gcd(i,2^r)}$. Therefore, $q^o$ has order $2^r$ if and only if $i$ is odd. On the other hand, if $i$ is even, then $q$ is a quadratic residue, and if $i$ is odd, then $q\equiv a^i=a \left(a^{\frac{i-1}{2}}\right)^2$ is a quadratic nonresidue. This implies that $q^o$ is a $2^r$-th root of unity modulo $p_{2^t}$ if and only if $q$ is a quadratic nonresidue modulo $p_{2^t}$.

    Ankeny shows in \cite{AnkenyNonResidue} that if the generalized Riemann hypothesis holds, there is a quadratic nonresidue in the first $\mathcal{O}\left((\log p_{2^t})^2\right)$ residue classes. For any residue $q$ it can be tested with $\mathcal{O}\left(\log p_{2^t}\right)$ multiplications and modulo operations whether $q^o$ has order $2^r$. As byproduct we get $\left(q^o\right)^{\left(2^{r-t}\right)}$. If and only if $q^o$ has order $2^r$, the power $\left(q^o\right)^{\left(2^{r-t}\right)}$ has order $2^t$.

    Therefore, a $2^t$-th root of unity modulo $p_{2^t}$ can be found in $\mathcal{O}\left((\log p_{2^t})^3\right)$ time.
\end{proof}

Therefore, we can efficiently compute the integer-convolution with the help of the number theoretic transform.

\begin{theorem}
    For a given integer $N$, we can find a modulus $p_N$ and a suitable root $q_N$ in $\mathcal{O}\left(N(\log N)^2 (\log\log N)\right)$ time such that it is possible to calculate the acyclic convolution modulo $p_N$ of two sequences of length $n\leq N$ in $\mathcal{O}\left(n\log n\right)$ time afterwards.
\end{theorem}
\begin{proof}
    The acyclic convolution of sequences of length $n$ can be derived from a cyclic convolution of sequences with lengths of at least $2n$. Therefore, it is sufficient to prepare $2^T$ with $2N\leq 2^T < 4N$.

    For this length, the last two theorems state that a suitable modulus $p_N$ and a suitable $2^T$-th root $q_N$ of unity can be found in $\mathcal{O}\left(N(\log N)^2 (\log\log N)\right)$.

    Afterwards, for every $n\leq N$ we can append zeros to get the length $2^t$ with $2n\leq 2^t < 4n$. Since $2^t$ is a divisor of $2^T$, we can use $\left(q_N\right)^{\left(2^{T-t}\right)}$ as $2^t$-th root of unity.

    This allows the calculation of the acyclic convolution modulo $p_N$ in $\mathcal{O}\left(n\log n\right)$ time.
\end{proof}

\newpage

\bibliography{cadences}

\begin{thebibliography}{10}

\bibitem{AgarwalNTTandFermat}
R.~C. Agarwal and C.~S. Burrus.
\newblock Number theoretic transforms to implement fast digital convolution.
\newblock {\em Proceedings of the IEEE}, 63(4):550--560, April 1975.

\bibitem{Cadences_Amir}
Amihood Amir, Alberto Apostolico, Travis Gagie, and Gad~M. Landau.
\newblock String cadences.
\newblock {\em Theoretical Computer Science}, 698:4--8, 2017.
\newblock Algorithms, Strings and Theoretical Approaches in the Big Data Era
  (In Honor of the 60th Birthday of Professor Raffaele Giancarlo).

\bibitem{AnkenyNonResidue}
N.~C. Ankeny.
\newblock The least quadratic non residue.
\newblock {\em Annals of Mathematics}, 55(1):65--72, 1952.

\bibitem{BachLinnikGRH}
Eric Bach and Jonathan Sorenson.
\newblock Explicit bounds for primes in residue classes.
\newblock {\em Math. Comput.}, 65(216):1717--1735, oct 1996.

\bibitem{Cormen}
Thomas~H. Cormen, Charles~E. Leiserson, Ronald~L. Rivest, and Clifford Stein.
\newblock {\em Introduction to Algorithms, 3rd Edition}.
\newblock {MIT} Press, 2009.

\bibitem{Cadences_Gardelle}
J.~Gardelle.
\newblock Cadences.
\newblock {\em Math\'ematiques et Sciences humaines}, 9:31--38, 1964.

\bibitem{LevcopoulosAlgo}
Christos Levcopoulos and Andrzej Lingas.
\newblock Bounds on the length of convex partitions of polygons.
\newblock In Mathai Joseph and Rudrapatna Shyamasundar, editors, {\em
  Foundations of Software Technology and Theoretical Computer Science}, pages
  279--295, Berlin, Heidelberg, 1984. Springer Berlin Heidelberg.

\bibitem{Linnik}
U.~V. Linnik.
\newblock On the least prime in an arithmetic progression. {I}. {T}he basic
  theorem.
\newblock {\em Rec. Math. [Mat. Sbornik] N.S.}, 15(57):139--178, 1944.

\bibitem{Cadences_Lothaire}
M.~Lothaire.
\newblock {\em Combinatorics on Words}.
\newblock Cambridge Mathematical Library. Cambridge University Press, 1997.

\bibitem{Mertens}
Franz Mertens.
\newblock Ein beitrag zur analytischen zahlentheorie.
\newblock {\em Journal f^^c3^^bcr die reine und angewandte Mathematik},
  78:46--62, 1874.

\bibitem{Schoenhage}
A.~Sch{\"o}nhage and V.~Strassen.
\newblock Schnelle multiplikation gro{\ss}er zahlen.
\newblock {\em Computing}, 7(3):281--292, Sep 1971.

\bibitem{ConvAppThompson}
William~B. Thompson, Peter Shirley, and James~A. Ferwerda.
\newblock A spatial post-processing algorithm for images of night scenes.
\newblock {\em Journal of Graphics Tools}, 7(1):1--12, 2002.

\bibitem{XylourisOld}
Triantafyllos Xylouris.
\newblock {\em \"{U}ber die {N}ullstellen der {D}irichletschen {L}-{F}unktionen
  und die kleinste {P}rimzahl in einer arithmetischen {P}rogression}, volume
  404 of {\em Bonner Mathematische Schriften [Bonn Mathematical Publications]}.
\newblock Universit\"{a}t Bonn, Mathematisches Institut, Bonn, 2011.
\newblock Dissertation for the degree of Doctor of Mathematics and Natural
  Sciences at the University of Bonn, Bonn, 2011.

\end{thebibliography}

\end{document}